\documentclass[12pt]{amsart}

\usepackage[T1]{fontenc}

\usepackage{amssymb}
\usepackage{array}
\usepackage[english]{babel}
\usepackage{nicefrac}
\usepackage{setspace} 
\usepackage[a4paper,tmargin=1in,bmargin=1.2in ,lmargin=1in,rmargin=1in,headheight=1in,headsep=1in,footskip=1.5\baselineskip]{geometry}

\usepackage{natbib}
\usepackage{amssymb}
\usepackage{fancyhdr}
\usepackage{bbm}
\usepackage{xcolor}
\usepackage{hyperref}
\hypersetup{
  colorlinks=true,
  linkcolor=blue,
  citecolor=blue
}
\usepackage{mleftright}
\usepackage{subcaption}
\usepackage{paralist}
\usepackage{comment}

\usepackage{caption}
\usepackage[labelfont=md]{subcaption}
\usepackage[normalem]{ulem}

\newtheorem{theorem}{Theorem}
\newtheorem*{theorem*}{Theorem}
\newtheorem{lemma}{Lemma}
\newtheorem{proposition}{Proposition}

\newtheorem{example}{Example}

\theoremstyle{definition}
\newtheorem{definition}{Definition}
\newtheorem*{definition*}{Definition}
\newtheorem{remark}{Remark}
\newtheorem*{lemma*}{Lemma}
\usepackage{graphicx}
\usepackage{enumerate}

\numberwithin{equation}{section}

\newcommand{\PP}{\mathbb{P}}
\newcommand{\EE}{\mathbb{E}}

\usepackage{xparse}
\DeclareDocumentCommand\Pr{ m g }{\ensuremath{
    {   \IfNoValueTF {#2}
      {\mathbb{P}\mleft[{#1}\mright]}
      {\mathbb{P}\mleft[{#1}\middle\vert{#2}\mright]}%
    }
}}
\DeclareDocumentCommand\E{ m g }{\ensuremath{
    {   \IfNoValueTF {#2}
      {\mathbb{E}\mleft[{#1}\mright]}
      {\mathbb{E}\mleft[{#1}\middle\vert{#2}\mright]}%
    }
}}

\usepackage[disable]{todonotes}
\usepackage{mathtools}

\usepackage{cleveref}

\usepackage{comment}
\usepackage{tikz-cd}
\usepackage{bm}
\usepackage[T1]{fontenc}
\usepackage{multirow}
\usepackage{pgfplots}
\usepackage{float}
\usepackage{subcaption}
\usepgfplotslibrary{fillbetween}
\pgfplotsset{compat=1.17}

\title{Dynamic Cheap Talk without Feedback}
\author{Atulya Jain}

 \thanks{Department of Economics, University of Bonn, Germany. Email: \href{mailto:ajain@uni-bonn.de}{ajain@uni-bonn.de}}

\thanks{I am indebted to  Tristan Tomala and Nicolas Vieille for their advice and encouragement throughout the project. I also thank  Sarah Auster, Francesc Dilme and Andreas Kleiner for  helpful suggestions and comments. I acknowledge funding from the DATAIA Convergence Institute (ANR-17-CONV-0003), a joint initiative of HEC Paris and Hi! PARIS, from the Programme d’Investissement d’Avenir (ANR-18-EURE-0005 / EUR DATA EFM), and from the German Research Foundation (DFG) under Germany’s Excellence Strategy – EXC 2126/2–390838866.}

\date{\today}

\begin{document}

\begin{abstract}
  We study a dynamic sender-receiver game in which the sender observes a state evolving according to a Markov chain but does not observe the receiver's action. Despite the absence of feedback, dynamic interaction partially restores commitment. We show that any equilibrium payoff of a persuasion model with partial commitment\textemdash where the sender can deviate to signaling policies that  preserve the marginal distribution over messages\textemdash can be achieved as a uniform equilibrium payoff in the dynamic game. Moreover, any convex combination of such payoffs across message distributions can also be attained. When the sender's payoff is state-independent, she achieves the Bayesian persuasion payoff.

  \medskip



    \end{abstract}

\maketitle

\vspace{-5mm}

\section{Introduction}

Communication between an informed expert and an uninformed decision-maker provides a natural framework for studying strategic information transmission. When their interests conflict, the expert may distort messages to steer the outcome in her favor, while the decision-maker acts based on his own interests. This limits the amount of information that can be transmitted and restricts the set of achievable outcomes. 

In many such relationships, decisions are made repeatedly over time. This dynamic interaction may help overcome this friction. When uncertainty evolves, information revealed today remains useful but decays over time, creating scope for intertemporal incentives. Such incentives are typically sustained through feedback, with future communication used to reward or punish behavior.

This carrot-and-stick argument, however, relies on the expert observing the decision-maker's actions. Yet, in many situations, the expert receives little or no feedback. For instance, a financial advisor may recommend portfolio allocations to an investor without observing whether they are followed, and a doctor may prescribe a treatment to a patient without observing compliance. Can credible communication arise without such feedback? If so, what are its limits?

We consider a dynamic sender-receiver  model in which the state evolves according to a Markov chain. In each period, the sender privately observes the state and sends a message to the receiver, who then chooses an action. The sender gets no feedback: she observes neither the receiver’s action nor her realized  payoff.  The players' stage  payoffs depend only on the current state and the receiver's action.

Despite the absence of feedback, dynamic interaction improves communication. Although the sender does not observe the receiver’s actions, the receiver observes the history of messages and can monitor their empirical frequencies. This
allows him to perform statistical tests on the distribution over messages and detect
deviations, thereby disciplining the sender’s behavior. As a result, the sender
faces endogenous restrictions on her deviations in the dynamic game, expanding the set of achievable outcomes beyond those in the one-shot game.

We focus on uniform equilibrium as the solution concept. A payoff vector is a uniform equilibrium payoff if, for every $\epsilon>0$, there is a strategy profile such that, for all sufficiently patient players, it is an $\epsilon$-equilibrium and induces a payoff vector within $\epsilon$ of the target. The key requirement is that, once $\epsilon$ is fixed, the same strategy profile works for all sufficiently high discount factors. This notion is well suited to our environment because players do not observe realized payoffs, which therefore cannot be used to discipline future play. It captures long-run incentives while allowing for arbitrarily small departures from exact optimality.


Our main result (\Cref{th:main}) shows that any equilibrium payoff of a persuasion model with partial commitment can be attained as a uniform equilibrium payoff in the dynamic game. In the static model, the receiver best responds to his posterior beliefs, while the sender can deviate only to signaling policies that preserve the marginal distribution over messages. Moreover, any convex combination of such payoffs across message distributions can also be attained. This static persuasion model is closely related to \citet{lin2024credible}, where the restriction on the sender’s deviations is imposed exogenously. Our result shows how this restriction emerges endogenously through dynamic interaction.

The intuition is as follows. Consider an equilibrium of the static  model with a given marginal distribution over messages. The game is played in blocks. In each period, the sender induces the same posterior beliefs as in the static equilibrium. This is feasible under our pseudo-renewal assumption, which ensures that the construction can be carried out recursively.\footnote{The equilibrium  construction also applies to general Markov chains, provided the induced posterior beliefs are closed under the Markov update; see \Cref{sec:pseudo-renewal}.} To enforce the target distribution over messages, the receiver monitors message frequencies within each block and follows the equilibrium response rule as long as these remain within their quotas. Otherwise, he plays the action corresponding to a different message whose quota has not yet been filled. This ensures that, within each block, the frequency of each message used matches its quota. For sufficiently long blocks, quotas bind only in a small fraction of periods toward the end of each block. As a result, the empirical distribution over states and actions can be made arbitrarily close to the target distribution. Consequently, the sender’s deviations in the dynamic game are effectively restricted to those that preserve the marginal distribution over messages. By the equilibrium conditions of the static model, the induced strategy profile is an approximate equilibrium of the dynamic game. Convex combinations are obtained by varying the target distributions across sub-blocks.

Thus, dynamic communication without feedback partially restores commitment, improving the sender's payoff relative to static cheap talk. In particular, when the sender's payoff is state-independent, she can fully bridge the commitment gap and achieve the Bayesian persuasion payoff, the upper bound in this setting.\footnote{This is also the sender's payoff upper bound in the dynamic game if she could commit to a dynamic signaling policy.} This contrasts with \citet{KUVALEKARetal}, who study the i.i.d.\ setting and show that, under Perfect Bayesian equilibrium, repetition cannot close a static commitment gap. The key distinction is that uniform equilibrium permits quota enforcement: message frequencies can be used to discipline the sender's deviations. This allows the sender to attain the Bayesian persuasion payoff in the state-independent case.

The previous result identifies a tractable subset of uniform equilibrium payoffs. However, in general, this is only a partial characterization. While the result builds on deviations that preserve the marginal distribution over messages, correlation across periods allows for richer statistical tests, making some such deviations detectable. \Cref{ex:counter-example} shows that the converse fails: uniform equilibrium payoffs need not lie in the convex hull of partial commitment payoffs. Nevertheless, Proposition~\ref{prop:necessary} provides a necessary condition: any uniform equilibrium payoff must be induced by an outcome that satisfies receiver obedience and is robust to a class of undetectable sender deviations based on permutations of the state process, following \citet{rsv}.  Finally, \Cref{prop:iid} shows that when states are i.i.d., all deviations that preserve the marginal distribution over messages are undetectable. Hence, the static  benchmark\textemdash the convex hull of equilibrium payoffs of the persuasion model with partial commitment\textemdash is exact, providing  a micro-foundation for  persuasion  with partial commitment.

\medskip

\textbf{Related Literature}: Our paper relates to the growing literature on strategic communication in dynamic environments where the underlying state evolves over time, see, e.g., \cite{rsv, escobar2013efficiency, renault2017optimal, ely2017beeps, margaria2018dynamic, meng2021value, Lehrer2026}. In particular, our paper contributes to the question of whether repeated interaction can substitute for commitment, see, e.g., \citet{KUVALEKARetal, pei2023repeated, mathevet2019reputation, jain2026calibratedforecastingpersuasion}. These papers offer rationales based on coarse summaries of past play, costs of lying, reputation, or statistical tests such as calibration. In contrast, we show that even without feedback, dynamic interaction can partially restore commitment through endogenous restrictions on the distribution over messages.

Our paper is closely related to \cite{KUVALEKARetal}, who also study repeated cheap talk without feedback. They consider an i.i.d.\ environment and characterize Perfect Bayesian equilibrium payoffs via an equivalence to static cheap talk with money burning. In contrast, we consider a Markovian setting and study uniform equilibrium payoffs. Under their solution concept, if the sender cannot attain the Bayesian persuasion payoff in static cheap talk, then she cannot attain it in the repeated game, even as players become patient. Moreover, when her preferences are state-independent, repetition does not improve her payoff. By contrast, we show that she can attain the Bayesian persuasion payoff precisely in this case. We discuss this distinction in \Cref{sec:uniform-pbe}.

Our equilibrium construction is based on quota mechanisms, where agents are required to match prescribed message frequencies; see, e.g., \cite{jackson2007overcoming,rsv,escobar2013efficiency,frankel2016discounted,ball2025quotamechanismsfinitesampleoptimality}. Our approach is closest to \citet{rsv}, where the receiver monitors message frequencies and replaces them when necessary to meet the prescribed quotas.

Our main result links the dynamic game to a static sender-receiver model, connecting our analysis to the literature on communication with partial commitment. In cheap talk \citep{Sobel, green2007two}, the sender cannot commit, while in Bayesian persuasion \citep{KamenicaBP}, she fully commits to a signaling policy. Our static benchmark lies in between: the sender can deviate to any policy that preserves the marginal distribution over messages. This restriction corresponds to the notion of credibility studied by \citet{lin2024credible}. In their model, the sender chooses the marginal distribution and the restriction is imposed exogenously. In contrast, the restriction arises endogenously in our dynamic game from monitoring message frequencies.\footnote{Another difference is that they focus on pure receiver strategies, in which case it is without loss to restrict attention to direct signaling policies.} The dynamic setting also allows convex combinations of equilibrium payoffs across message distributions. More broadly, our formulation relates to work that weakens commitment power in Bayesian persuasion; see, for instance, \citet{min2021Bayesian} and \citet{lipnowski2022persuasion}. 

\medskip

The rest of the paper is organized as follows. In \Cref{sec:Model}, we introduce the dynamic game without feedback and discuss the value of commitment. In \Cref{Sec:maintresults}, we present the persuasion model with partial commitment and provide our main results.  In \Cref{sec:discussion}, we discuss the role of uniform equilibrium and the pseudo-renewal assumption.  In \Cref{sec:conclusion}, we conclude and discuss future work. All omitted proofs are in \Cref{proofs}.

\section{Model}
\label{sec:Model}

We consider a dynamic game  between  two players: the sender (she) and the receiver (he). The state of the world evolves over time in a discrete-time infinite horizon. In each period $n \in \mathbb{N}$,  the sender privately observes the current state $\omega_n \in \Omega$ and sends a message $m_n \in M$ to the receiver. Upon observing the message, the receiver chooses an action $a_n \in A$.  The action is not disclosed to the sender.  The sender's and receiver's payoffs in that period are \(u_S(\omega_n,a_n)\) and \(u_R(\omega_n,a_n)\), respectively.\footnote{Neither player observes realized payoffs. In particular, the sender does not observe the action \(a_n\), while the receiver does not observe the state \(\omega_n\).}

The state evolves according to a Markov chain.   Given the current state $\omega_n$, the state in the next period $\omega_{n+1}$ is drawn  according to the transition matrix $Q(\omega_{n+1}\mid \omega_n)$. We assume the Markov chain   is irreducible and pseudo-renewal.

\begin{definition}
A Markov chain $\{\omega_n\}_{n \in \mathbb{N}}$ is \textbf{pseudo-renewal} if, for all $\omega \neq \tilde{\omega}$, $Q(\omega \mid \tilde{\omega}) = \beta_\omega$, where $\{\beta_\omega\}_{\omega \in \Omega}$ are non-negative numbers satisfying $\sum_{\omega} \beta_\omega \le 1$.
\end{definition}

Under this assumption, for any two distinct states, the transition probability depends only on the destination state and not on the origin state. This condition includes i.i.d.\ processes,  binary-state Markov chains, and  shock-driven processes in which shocks are i.i.d.\ and the state remains constant between shocks.\footnote{\citet{rsv,renault2017optimal} and \citet{horner2017keeping} also study  games with pseudo-renewal chains. Under pseudo-renewal, irreducibility is equivalent to $\beta_\omega>0$ for every $\omega$, which in turn implies aperiodicity.} The assumption is restrictive, but it is used only as a sufficient condition for the implementation. In \Cref{sec:pseudo-renewal}, we return to its role and provide a condition under which the main result extends to arbitrary Markov chains.

Let \(\mu \in \Delta(\Omega)\) denote the unique invariant distribution of the Markov chain. The initial state \(\omega_1\) is drawn according to \(\mu\). The sets of states \(\Omega\), actions \(A\), and messages \(M\) are all finite. We assume that \(|M| \geq \min \{ |\Omega|, |A| \}\).

The players' strategies map their information into choices. For the sender,  a strategy is a sequence of mappings $\sigma=(\sigma_n)_{n\ge1}$, where $\sigma_n:(\Omega \times M)^{n-1}\times \Omega \to \Delta(M)$ specifies a distribution over messages given past states and messages and the current state. 
For the receiver, a strategy is a sequence of mappings $\tau=(\tau_n)_{n\ge1}$, where $\tau_n: (M \times A)^{n-1} \times M \to \Delta (A)$ specifies a distribution over actions given past messages and actions and the current message. Given a strategy profile \((\sigma,\tau)\) and \(\delta \in (0,1)\), the \(\delta\)-discounted payoff of player \(i \in \{S,R\}\) is
 \begin{equation*}
  \gamma^\delta_i(\sigma, \tau)=\mathbb{E}_{\sigma, \tau}\Big[(1-\delta)\sum_{n=1}^\infty \delta^{n-1} u_i(\omega_n, a_n)\Big].
\end{equation*}



\textbf{Uniform equilibrium:} We focus on uniform equilibrium,  a standard solution concept in stochastic and repeated games, especially in settings with imperfect monitoring.\footnote{See, e.g., \citet{Lehrer1990,fudenberg1991approximate,tomala1999nash,vieille2000two,mertens2015repeated,deb2016uniform,solan2022course}.} A payoff vector is a uniform equilibrium payoff if, for every $\epsilon>0$, there exists a single strategy profile that satisfies two requirements for all sufficiently patient players: no player can gain more than $\epsilon$ by deviating, and the induced payoff is within $\epsilon$ of the payoff vector.  This notion is well suited to our environment, where players do not observe realized payoffs and therefore cannot use them to discipline future play. Rather than imposing sequential rationality after every history, it captures long-run behavior up to arbitrarily small departures from exact optimality.

\begin{definition}
A payoff vector $\boldsymbol{\gamma}=(\gamma_S,\gamma_R)\in \mathbb{R}^2$ is a \textbf{uniform equilibrium payoff} if for every $\epsilon>0$, there exist $\delta_0\in(0,1)$ and a strategy profile $(\sigma^*,\tau^*)$ such that for all $\delta \ge \delta_0$,
\begin{align*}
\gamma^\delta_S(\tilde{\sigma},\tau^*) &\le \gamma^\delta_S(\sigma^*,\tau^*)+\epsilon \quad \forall \tilde{\sigma},\\
\gamma^\delta_R(\sigma^*,\tilde{\tau}) &\le \gamma^\delta_R(\sigma^*,\tau^*)+\epsilon \quad \forall \tilde{\tau},
\end{align*}
and
\begin{equation*}
{\|\boldsymbol{\gamma}^\delta(\sigma^*,\tau^*)-\boldsymbol{\gamma}\|}_{\infty} \le \epsilon.
\end{equation*}
\end{definition}
 Let $\mathcal{U}$ denote the set of uniform equilibrium payoffs.

 \medskip

\textbf{Value of Commitment:} We show that the sender's static commitment payoff at the invariant distribution, the Bayesian persuasion payoff, is an upper bound on her equilibrium payoff in the dynamic game. To show this, we adopt a belief-based representation of the dynamic game, under which payoffs can be tracked through the receiver's beliefs.

In the dynamic game, the receiver updates his belief over the states given the sender's message. Since the state follows a Markov chain, beliefs evolve across periods.   Given any history, let $q_n$ and $p_n$ denote the receiver's belief about the state $\omega_n$ before and after observing the message in period $n$, respectively. The transition matrix $Q$ links the two beliefs:
\[
q_{n+1}(\omega)=\sum_{\tilde \omega \in \Omega} p_n(\tilde\omega)\, Q(\omega\mid \tilde \omega).
\]

Since the sender does not observe the receiver’s actions, the current action does not affect the distribution of future play. Hence, a myopic best response to his posterior belief is optimal in each period. Let $\tau^{\mathrm{myp}}$ denote the receiver’s strategy that, at every history, chooses a myopic best response to his posterior belief. In any $\epsilon$-equilibrium, the receiver therefore follows a myopic best response in all but a small fraction of periods.

Let
\[
\hat{u}_S(p) := \mathbb{E}_p[u_S(\omega, a^*(p))]
\]
denote the sender's indirect utility, where $a^*(p)$ is the receiver's best response given belief $p$.\footnote{If multiple actions are optimal, we assume ties are broken in favor of the sender, as is standard in the persuasion literature.}

In any period, the sender's strategy induces a Bayes plausible distribution over posterior beliefs, i.e., the expected posterior equals the prior. In the static persuasion model, \citet{KamenicaBP} show that the sender's maximal payoff, given prior \(\pi\), is the concave envelope of the sender's indirect utility, denoted by \(\mathrm{Cav}\,\hat{u}_S(\pi)\).

The same reasoning applies period by period in the dynamic game.  In period \(n\), the sender induces a Bayes plausible distribution over posterior beliefs with prior \(q_n\), so her expected payoff is at most \(\mathrm{Cav}\,\hat{u}_S(q_n)\). Therefore, against \(\tau^{\mathrm{myp}}\), the sender's
\(\delta\)-discounted payoff under any strategy \(\sigma\) satisfies
\begin{align*}
\gamma_S^\delta(\sigma,\tau^{\mathrm{myp}})
&\le (1-\delta)\sum_{n=1}^\infty \delta^{n-1}\mathbb{E}_{\sigma}\!\left[\mathrm{Cav}\,\hat{u}_S(q_n)\right] \\
&\le (1-\delta)\sum_{n=1}^\infty \delta^{n-1}\mathrm{Cav}\,\hat{u}_S(\mathbb{E}_{\sigma}[q_n]) \\
&= \mathrm{Cav}\,\hat{u}_S(\mu), 
\end{align*}
where the second inequality follows from Jensen's inequality, and the last equality uses \(\mathbb{E}_\sigma[q_n]=\mu\) for all \(n\). Thus, against $\tau^{\mathrm{myp}}$, the sender's payoff is at most the Bayesian persuasion payoff at $\mu$. As we will show in \Cref{prop:necessary}, every uniform equilibrium payoff is induced by an obedient outcome, so the same bound applies to the sender's payoff in any uniform equilibrium.

The stage game corresponds to cheap talk, where the sender cannot commit to how messages are generated.\footnote{Any cheap talk equilibrium payoff with prior $\mu$ is achievable as a uniform equilibrium payoff.} When the sender's preferences are state-independent, \citet{lipnowski2020cheap} show that her maximal cheap talk payoff is given by the quasi-concave envelope of \(\hat{u}_S\). The gap between the concave and quasi-concave envelopes captures the value of commitment (see \Cref{fig:BP-CT}). 

Our objective is to characterize how much of this gap can be recovered through dynamic interaction without feedback. We show that the sender's equilibrium payoff improves upon cheap talk, partially restoring commitment. In particular, when the sender's preferences are state-independent, she attains the Bayesian persuasion payoff.

\section{Main Results}
\label{Sec:maintresults}

 Our main result partially characterizes the set of uniform equilibrium payoffs in terms of a static persuasion model with partial commitment. We first define the static model. We then show that any equilibrium payoff of the static model, as well as any convex combination of such payoffs, can be sustained as a uniform equilibrium payoff. Finally, we show that uniform equilibrium payoffs may lie outside this set and provide a necessary condition.

\subsection{Persuasion with partial commitment}
The key benchmark in our analysis is a static persuasion model with partial commitment. The environment coincides with the stage game of the dynamic model: the sets of states \(\Omega\), messages \(M\), and actions \(A\), as well as payoff functions \(u_S\) and \(u_R\), are the same.

Fix a prior \(\pi \in \Delta(\Omega)\) and a marginal distribution over messages \(\lambda \in \Delta(M)\). The state \(\omega\) is drawn according to \(\pi\). The sender observes \(\omega\) and chooses a message $m$ according to a signaling policy \(\rho:\Omega\to\Delta(M)\) satisfying
\[
\sum_{\omega\in\Omega}\pi(\omega)\rho(m\mid\omega)=\lambda(m)
\quad \text{for all } m\in M.
\]
After observing the message, the receiver chooses an action $a$ via a response rule \(\kappa:M\to\Delta(A)\). An equilibrium requires that the receiver best responds to posterior beliefs and that  the sender cannot profitably deviate within the set of policies that preserve the marginal  distribution \(\lambda\). A babbling equilibrium always exists, in which messages are uninformative and the receiver  best responds to the prior.

The static benchmark imposes a fixed marginal distribution over messages: the sender may deviate only to signaling policies that preserve $\lambda$. In the dynamic game, the same restriction arises endogenously through the receiver's monitoring of message frequencies. 

We formulate equilibrium directly in terms of the posterior beliefs  \(\{p_m\}_{m\in M}\) and the receiver's response rule $\kappa:M \to \Delta(A)$. Let \(\Sigma(\pi,\lambda)\) denote the set of Bayes plausible collections of posterior beliefs  with weights \(\lambda\), that is, \(\{p_m\}_{m\in M}\in \Sigma(\pi,\lambda)\) if $\sum_{m}\lambda(m)p_m=\pi$.

Given \(\{p_m\}_{m\in M}\in \Sigma(\pi,\lambda)\) and a response rule \(\kappa: M \to \Delta(A)\), denote the induced joint distribution over states and actions by \(\nu \in \Delta(\Omega\times A)\), where
\[
\nu(\omega,a)=\sum_{m\in M}\lambda(m)p_m(\omega)\kappa(a\mid m).
\]
The expected payoff of player \(i\in\{S,R\}\) is
\[
\sum_{m\in M}\lambda(m)\sum_{\omega\in\Omega}p_m(\omega)\sum_{a\in A}\kappa(a\mid m)u_i(\omega,a).
\]

\medskip

 The equilibrium conditions are as follows.

\textbf{Receiver.} The receiver best responds to his posterior beliefs: for every 
\(m \in M\) with \(\lambda(m)>0\),
\begin{equation}\label{Eq-R} \tag{Eq-R} 
\mathrm{Supp} \,\kappa(\cdot \mid m)\subseteq A^*(p_m),
\end{equation}
where \(A^*(p)=\arg\max_{b \in A} \sum_{\omega} p(\omega)u_R(\omega,b)\) denotes the receiver's best response correspondence at belief $p$.

\textbf{Sender.} The sender cannot profitably deviate within the set of Bayes plausible collections of posterior beliefs with weights \(\lambda\). That is, for any alternative 
\(\{\tilde p_m\}_{m\in M} \in \Sigma(\pi,\lambda)\), 
\begin{equation}\label{Eq-S} \tag{Eq-S} 
\sum_{m\in M}\lambda(m)\sum_{\omega \in\Omega}
\Bigl(p_m(\omega)-\tilde p_m(\omega)\Bigr) \sum_{a \in A} \kappa(a \mid m)u_S(\omega,a) \geq 0.
\end{equation}


Let \(\mathcal E(\pi,\lambda)\) denote the set of equilibrium payoffs under prior \(\pi\) and marginal \(\lambda\).

To illustrate the persuasion model with partial commitment, consider the following example, adapted from \cite{lipnowski2020cheap}.

\begin{example}\label{example}
A political think tank (sender) advises a lawmaker (receiver) on whether to implement one of two reforms, denoted \(1\) and \(2\), or to implement no reform. The corresponding action set is \(A=\{a_1,a_2,a_0\}\). The think tank's payoff is state-independent: it prefers reform 2 to reform 1, and any reform to no reform. The lawmaker prefers to implement the reform that is good and to keep the status quo if neither reform is sufficiently likely to be good. Let  \(\Omega=\{\omega_1,\omega_2\}\), where \(\omega_i\) is the state in which reform \(i\) is good, and let \(M=\{m_1,m_2\}\) denote the message set. The players' payoffs are given by the following matrix:
  \begin{table}[h!]
  \centering
    \setlength{\extrarowheight}{2pt}
    \begin{tabular}{cc|c|c|c|}
      & \multicolumn{1}{c}{} & \multicolumn{1}{c}{$a_2$}  & \multicolumn{1}{c}{$a_1$} & \multicolumn{1}{c}{$a_0$}\\\cline{3-5}
      \multirow{2}*{}  & $\omega_1$ & $(2,0)$ & $(1,4)$ &  $(0,3)$ \\\cline{3-5}
      & $\omega_2$ & $(2,4)$ & $(1,0)$ & $(0,3)$ \\\cline{3-5}
    \end{tabular}
  \end{table}
  
Since the state space is binary, we identify beliefs with \(p=\PP(\omega_2)\in[0,1]\). The receiver's best response (ties broken in favor of the sender) is
\[
a^*(p)=
\begin{cases}
a_1 & \text{if } p \le \tfrac{1}{4},\\
a_0 & \text{if } \tfrac{1}{4} < p < \tfrac{3}{4},\\
a_2 & \text{if } p \ge \tfrac{3}{4}.
\end{cases}
\]

Under Bayesian persuasion and cheap talk, the sender's maximal payoff is given by the concave and quasi-concave envelopes of the sender's indirect utility, evaluated at the prior (see \Cref{fig:BP-CT}). Their difference captures the value of commitment.

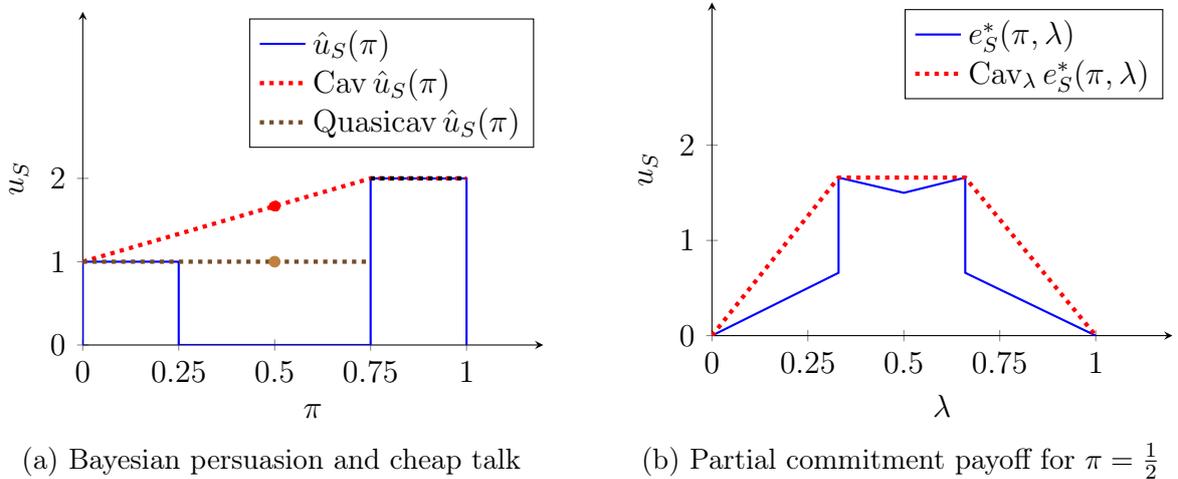
\begin{figure}[htbp]
\centering

\begin{subfigure}{0.48\textwidth}
\centering
\begin{tikzpicture}
\begin{axis}[
width=\linewidth,
height=6cm,
xlabel=$\pi$,
ylabel=$u_S$,
axis x line=bottom,
axis y line=left,
xmin=0, xmax=1.2,
ymin=0, ymax=4,
xtick={0,0.25,0.5,0.75,1},
ytick={0,1,2},
legend cell align=left,
legend style={
    at={(0.98,0.98)},
    anchor=north east,
    fill=white,
    draw=black
}
]

\addplot+[const plot, no marks, thick]
coordinates {(0,0) (0,1) (0.25,1) (0.25,0) (0.75,0) (0.75,2) (1,2) (1,0)};
\addlegendentry{$\hat{u}_S(\pi)$}

\addplot+[dotted, no marks, ultra thick]
coordinates { (0,1) (0,1) (0.75,2) (1,2)};
\addlegendentry{$\mathrm{Cav}\,\hat{u}_S(\pi)$}

\addplot+[const plot, dotted, no marks, ultra thick]
coordinates { (0,1) (0.25,1) (0.25,1) (0.75,1)};
\addlegendentry{$\mathrm{Quasicav}\,\hat{u}_S(\pi)$}

\addplot+[const plot, dotted, no marks, ultra thick]
coordinates {(0.75,2) (1,2)};

\addplot+[
only marks,
mark=*,
mark size=2pt,
mark options={fill=brown, draw=brown}
] coordinates {(0.5,1)};

\addplot+[
only marks,
mark=*,
mark size=2pt,
mark options={fill=red, draw=red}
] coordinates {(0.5,5/3)};

\end{axis}
\end{tikzpicture}
\caption{Bayesian persuasion and cheap talk}
\label{fig:BP-CT}
\end{subfigure}
\hfill
\begin{subfigure}{0.48\textwidth}
\centering
\begin{tikzpicture}
\begin{axis}[
width=\linewidth,
height=6cm,
xlabel=$\lambda$,
ylabel=$u_S$,
axis x line=bottom,
axis y line=left,
xmin=0, xmax=1.2,
ymin=0, ymax=3.5,
xtick={0,0.25,0.5,0.75,1},
ytick={0,1,2},
legend cell align=left,
legend style={
    at={(0.98,0.98)},
    anchor=north east,
    fill=white,
    draw=black
}
]

\addplot+[name path=lower, no marks, thick]
coordinates {(0,0) (0.33,0.66) (0.33,1.66) (0.5,1.5) (0.66,1.66) (0.66,0.66) (1,0)};
\addlegendentry{$e^*_S(\pi,\lambda)$}

\addplot+[name path=upper, ultra thick, dotted, no marks]
coordinates {(0,0) (0.33,1.66) (0.66,1.66) (1,0)};
\addlegendentry{$\mathrm{Cav}_\lambda\,e^*_S(\pi,\lambda)$}


\end{axis}
\end{tikzpicture}
\caption{Partial commitment payoff for $\pi=\frac{1}{2}$}
\label{fig:partial-commitment}
\end{subfigure}

\caption{Comparison of benchmarks and partial commitment}
\label{fig:combined}
\end{figure}

Consider persuasion with partial commitment. The marginal \(\lambda \in \Delta(M)\) restricts the set of feasible collections of posterior beliefs. Since \(M=\{m_1,m_2\}\), we identify \(\lambda=\PP(m_1)\in[0,1]\). For instance, let \(\pi=1/2\) and \(\lambda=1/2\). Then any feasible pair of posteriors must satisfy \(p_{m_1}=1/2 + x\) and \(p_{m_2}=1/2 - x\) for some \(x \in [-1/2,\,1/2]\). The sender's optimal equilibrium is to induce \(p_{m_1}=0\) and \(p_{m_2}=1\), following which the receiver takes actions \(a_1\) and \(a_2\), respectively. To see that this constitutes an equilibrium, note that the receiver best responds to these beliefs, and any sender's deviation that preserves \(\lambda\) does not affect her payoff. The sender's expected payoff is \(3/2\), strictly below the Bayesian persuasion payoff of \( 5/3\).

\Cref{fig:partial-commitment} plots the sender's maximal equilibrium payoff \(e_S^*(\pi,\lambda)\) as a function of \(\lambda\) (blue solid line). In particular, when \(\lambda=1/3\), the sender attains the Bayesian persuasion payoff. Moreover, convex combinations across marginals are attainable in the dynamic game. In particular, the sender can achieve the concave envelope $\mathrm{Cav}_\lambda\, e_S^*(\pi,\lambda)$ (red dotted line),  where
\(\operatorname{Cav}_{\lambda}\) denotes concavification with respect to the
marginal distribution \(\lambda\). This shows that allowing for variation across message distributions expands the set of attainable payoffs and enables the sender to achieve strictly higher payoffs than under some fixed marginals.
\end{example}

This model lies between Bayesian persuasion and cheap talk. In Bayesian persuasion, the sender faces no restriction on the marginal distribution over messages and cannot deviate after choosing a signaling policy, whereas in cheap talk she can deviate freely. Here, deviations are restricted to policies that preserve the marginal distribution over messages. This restriction is closely related to the notion of credibility studied by \citet{lin2024credible}, who characterize sender incentive compatibility through cyclic monotonicity conditions. The difference is that we fix the marginal distribution over messages and show how the restriction arises endogenously in the dynamic game from monitoring message frequencies.

Let \(e^*(\pi,\lambda)  \in \mathcal{E}(\pi,\lambda)\) denote a sender-optimal equilibrium payoff vector, and define 
\[
e_S^*(\pi):=\max_{\lambda \in \Delta(M)} e_S^*(\pi,\lambda).
\]

\begin{remark}\label{remark:cheap-talk}
 Any cheap talk equilibrium is also an equilibrium of the persuasion model with partial commitment under the induced marginal distribution over messages.
\end{remark}

The receiver's incentive constraint is unchanged, while cheap talk allows a larger set of deviations than under partial commitment.

Let \(u_S^{\mathrm{CT}}(\pi)\) and \(u_S^{\mathrm{BP}}(\pi)=\mathrm{Cav}\,\hat{u}_S(\pi)\) denote the sender's maximal equilibrium payoff in cheap talk and Bayesian persuasion, respectively. Then
\[
u_S^{\mathrm{CT}}(\pi) \le e_S^*(\pi) \le u_S^{\mathrm{BP}}(\pi).
\]

\begin{remark}\label{remark:benchmark}
If the sender's payoff is state-independent, then $e_S^*(\pi)= u_S^{\mathrm{BP}}(\pi)$.
\end{remark}

Indeed, when the sender's payoff is state-independent, deviations that preserve the distribution over messages do not change the distribution over actions and therefore do not affect her payoff. Hence, the sender attains the Bayesian persuasion payoff.

\subsection{Uniform equilibrium payoffs}
We now state the main result. It establishes that any equilibrium payoff of the  persuasion model with partial commitment, evaluated at the invariant distribution $\mu$, as well as any convex combination thereof, can be attained as a uniform equilibrium payoff.  Let $\mathcal{E}(\pi):=\bigcup_{\lambda \in \Delta (M)} \mathcal{E}(\pi,\lambda)$ denote the set of equilibrium payoffs under prior \(\pi\) for some marginal distribution \(\lambda\).  

\begin{theorem}\label{th:main}
Any equilibrium payoff of the persuasion model with partial commitment at prior $\mu$, and any convex combination of such payoffs across message distributions, is a uniform equilibrium payoff. Equivalently,
\begin{equation*}
\mathrm{Co} \, \mathcal{E}(\mu) \subseteq \mathcal{U}. 
\end{equation*}
\end{theorem}

\textbf{Proof Sketch:} We first show that any  equilibrium payoff $e \in \mathcal{E}(\mu)$ of  the persuasion model with partial commitment is a uniform equilibrium payoff.   We then extend the argument to convex combinations.

Fix \(\lambda \in \Delta(M)\) and an equilibrium payoff \(e \in \mathcal{E}(\mu,\lambda)\). 
Let \(\{p_m\}_{m \in M} \in \Sigma(\mu,\lambda)\) and 
\(\kappa: M \to \Delta(A)\) denote posterior beliefs and a response rule that induce \(e\).

We construct a strategy profile $(\sigma^*,\tau^*)$. Play is divided into consecutive blocks of length $N$.\footnote{Block strategies are commonly used in the analysis of uniform equilibria under imperfect monitoring to detect deviations via statistical tests; see, e.g., \cite{Lehrer1990,lehrer1992two,fudenberg1991approximate,tomala1999nash,deb2016uniform}.} At the start of each block, only the receiver's quota counters are reset; prescribed play otherwise continues according to the strategies described below. For any rational $\lambda$, choose $N$ such that $N \lambda(m)$ is an integer for every $m \in M$.\footnote{For irrational $\lambda$, the result follows by approximation with rational distributions.}

\medskip

\textbf{Sender's strategy}: The sender's strategy is constructed to induce the same posterior beliefs ${\{p_m\}}_{m\in M}$ in every period.  This ensures that, at every period, the  joint distribution over states and messages matches the corresponding static outcome.

For each \(m\in M\), let \(q_m=p_mQ\) denote the receiver's prior belief in the next period
following message \(m\).  Under the pseudo-renewal and the irreducibility assumption, there exists $\alpha \in [0,1)$ such that for every $m$,\footnote{Indeed, pseudo-renewal implies that \(pQ=\alpha p+\beta\), where
\(\alpha=1-\sum_{\omega\in\Omega}\beta_\omega\) and
\(\beta=(\beta_\omega)_{\omega\in\Omega}\). Since \(\mu\) is invariant,
\(\beta=(1-\alpha)\mu\).}
\[
q_m = \alpha p_m + (1-\alpha)\mu= \alpha p_m + (1-\alpha)\sum_{m'\in M} \lambda(m')p_{m'}.
\]
Thus, after message \(m\), the next-period prior \(q_m\) can be written as a convex combination of the current posterior \(p_m\) and the invariant prior \(\mu\). The sender uses this decomposition to generate the target posteriors.\footnote{The construction also applies to arbitrary Markov chains under a closure condition on the induced posteriors; see \Cref{sec:pseudo-renewal}.} To do this, the sender uses a simple recursive strategy: in each period, the message depends only on the current state and the previous message.\footnote{For $n=1$, the sender uses the equilibrium strategy of the static model.}
\[
\sigma_n^*(m_n=m \mid \omega_n=\omega, m_{n-1}=m')
=
\frac{p_{m}(\omega)}{q_{m'}(\omega)}
\Big((1-\alpha) \lambda(m) + \alpha \mathbf{1}_{\{m = m'\}}\Big).
\]

By Bayes' rule, this strategy induces posterior belief \(p_m\) after each message \(m\). The induced evolution of beliefs is illustrated in \Cref{fig:senderstrategy}.

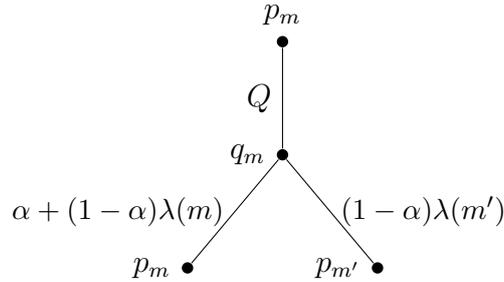
\begin{figure}[htp!]
  \centering
  \begin{tikzpicture}[thin,
    level 1/.style={sibling distance=40mm},
    level 2/.style={sibling distance=25mm},
    level 3/.style={sibling distance=15mm},
    every circle node/.style={minimum size=1.5mm,inner sep=0mm}]
    
    \node[circle,fill,label=above:$p_m$] (root) { }
      child { node [circle,fill,label=left:$q_m$] (A) { }
        child {node [circle,fill,label=left:$p_m$] (B) { }
          edge from parent
          node[left] {\small $\alpha + (1-\alpha)\lambda(m)$}}
        child {node [circle,fill,label=left:$p_{m'}$] (D) { }
          edge from parent
          node[right] {\small $(1-\alpha)\lambda(m')$}}
        edge from parent
        node[left] {$Q$}};
  \end{tikzpicture}
  \caption{Recursive decomposition of receiver's beliefs under $\sigma^*$}
  \label{fig:senderstrategy}
\end{figure}

\textbf{Receiver's strategy}: In each period \(n\), the sender sends a message \(m_n\). If \(m_n\) has not yet reached its quota, the receiver plays according to \(\kappa(\cdot \mid m_n)\), a best response given his posterior belief. Otherwise, he ignores \(m_n\) and instead plays according to \(\kappa(\cdot \mid m')\) for some different message \(m'\) whose quota has not yet been exhausted. Thus, the message sent by the sender need not coincide with the message used by the receiver. Note that $\tau^*$ need not be sequentially rational at every history: when a quota binds, the receiver may have to take an action that is not his myopic best response.\footnote{The construction is not a Perfect Bayesian equilibrium. We discuss this distinction and its implications in \Cref{sec:uniform-pbe}.}

To track message frequencies within a block, let $r_n \in M$ denote the message used by the receiver in period $n$, and define
\[
C_n(m)
:=
\sum_{k=1}^n \mathbf{1}_{\{r_k = m\}}
\]
as the number of times message \(m\) has been used up to period \(n\). Under \(\tau^*\), the receiver plays according to \(\kappa(\cdot \mid r_n)\). If \(C_{n-1}(m_n) < N\lambda(m_n)\), then the receiver follows the sender's message, so \(r_n = m_n\). Otherwise, he chooses a different message \(r_n = m' \neq m_n\) such that \(C_{n-1}(m') < N\lambda(m')\). For example, \(m'\) can be chosen randomly with probability proportional to the remaining quota \(N\lambda(m') - C_{n-1}(m')\). Hence, in each block, each message \(m\) is used exactly \(N\lambda(m)\) times.

The proof proceeds in two steps. First, we show that the payoff induced by $(\sigma^*,\tau^*)$ can be made arbitrarily close to $e$ for sufficiently patient players. Second, we show that, for any $\epsilon>0$, the profile $(\sigma^*,\tau^*)$ is an $\epsilon$-equilibrium  for sufficiently patient players.\footnote{The strategy profile is also a uniform equilibrium with full feedback. This relates to Proposition 2 in \cite{rsv}, which shows that Nash equilibrium payoffs without feedback are a subset of those with full feedback.}

Under $\sigma^*$, the expected joint distribution over states and messages in any period $n$ coincides with the static distribution, i.e.,
\[
\mathbb{E}_{\sigma^*}\big[\mathbf{1}_{\{\omega_n=\omega,m_n=m\}}\big] = \lambda(m)p_m(\omega).
\]

Hence, for a large block $N$, the receiver acts according to the sender's message in most periods with high probability. Therefore, by choosing $N$ large enough and considering players  sufficiently patient, the induced payoff can be made arbitrarily close to $e$.

Now, we show that for any $\epsilon>0$, the strategy profile $(\sigma^*,\tau^*)$ is an $\epsilon$-equilibrium for all sufficiently patient players.

First, consider the receiver’s deviations. Since the sender does not observe the receiver’s actions, his current choice does not affect future play, so a myopic best response is optimal at every history. Under \(\tau^*\), whenever the quota is not exhausted, the receiver responds to \(m_n\) according to \(\kappa(\cdot \mid m_n)\), which is optimal given his posterior. For sufficiently large \(N\), this holds in all but a small fraction of periods. Hence, \(\tau^*\) coincides with a myopic best response in almost all periods, so any gain from deviation vanishes for sufficiently patient players.

Next, consider the sender's deviations. The receiver's strategy enforces that each message $m$ is used exactly $N\lambda(m)$ times within a block. Hence, any  deviation $\tilde{\sigma}$ induces a collection of beliefs $\{ \tilde p_m \}_{m \in M}$, where for each $m$ with $\lambda(m)>0$,
\[
\tilde p_m(\omega)
:=
\frac{1}{N\lambda(m)} \EE_{\tilde{\sigma},\tau^*} \Bigg[\sum_{n=1}^N \mathbf{1}_{\{\omega_n=\omega,\, r_n=m\}}\Bigg].
\]
Here, $\tilde p_m$ is the empirical distribution of states in periods in which the receiver uses message $m$.

By construction, $\{\tilde p_m\}_{m \in M} \in \Sigma(\mu,\lambda)$ and therefore defines a feasible deviation in the static model. By the sender's equilibrium condition, no such deviation yields a higher payoff.  Moreover, for sufficiently large $N$ and sufficiently patient players, the sender's discounted payoff under $\tilde \sigma$ is arbitrarily close to the payoff of the corresponding static deviation.  Hence, $e\in \mathcal{U}$.

We now extend the previous block construction to obtain convex combinations of equilibrium payoffs.\footnote{Since payoffs lie in \(\mathbb R^2\), Carathéodory's theorem implies that finite convex combinations are without loss.} Let $\lambda=\sum_{i=1}^n \eta_i \lambda_i$ with $\eta_i \ge 0$ and $\sum_{i=1}^n \eta_i=1$ and let $e_i \in \mathcal{E}(\mu,\lambda_i)$. The players still use a block strategy. Choose \(N\) such that \(N \eta_i \lambda_i(m)\) is an integer for each \(i\) and \(m \in M\). Each block of length $N$ is then divided into sub-blocks of lengths $\{N \eta_i\}_{i=1}^n$, and in sub-block $i$ the players play according to the strategy profile associated with $\lambda_i$. Between sub-blocks, we insert short transition phases with negligible discounted weight, so that beliefs can return close to the relevant starting beliefs.

Using the same argument as before, each sub-block yields an average payoff arbitrarily close to $e_i$ for sufficiently large $N$ and sufficiently patient players. Since sub-block $i$ occupies a fraction $\eta_i$ of each block, the overall payoff is arbitrarily close to $\sum_{i=1}^n \eta_i e_i$. Moreover, because the quota counters are reset at the beginning of each sub-block, deviations can be evaluated separately within each sub-block, and the same incentive argument applies.

\subsection{Failure of the Converse}
Theorem \ref{th:main} shows that all convex combinations of equilibrium payoffs of the persuasion model with partial commitment are uniform equilibrium payoffs. We show that the converse inclusion does not hold.

In the static model with partial commitment, the sender can deviate to any policy that preserves the marginal distribution over messages. In the dynamic game, however, not all such deviations are undetectable, as the receiver can condition on richer statistical tests, such as frequencies of message pairs. The following example shows that some uniform equilibrium payoffs lie outside the convex hull of equilibrium payoffs in the persuasion model with partial commitment.

\begin{example}\label{ex:counter-example}
Consider binary states $\Omega=\{\omega_1,\omega_2\}$, binary messages
$M=\{\ell,h\}$, and actions $A=\{a_1,a_2,a_0,a_3,a_4\}.$
Since the state space is binary, we identify a belief with
$p=\PP(\omega_2)\in[0,1]$. Consider the transition matrix
\[
Q=
\begin{bmatrix}
\frac23 & \frac13\\[1mm]
\frac13 & \frac23
\end{bmatrix},
\]
which has invariant distribution $\mu=\frac12$.

The sender's and receiver's payoffs are given by
\begin{table}[htp!]
\centering
\setlength{\extrarowheight}{2pt}
\begin{tabular}{cc|c|c|c|c|c|}
 & \multicolumn{1}{c}{}
 & \multicolumn{1}{c}{$a_1$}
 & \multicolumn{1}{c}{$a_2$}
 & \multicolumn{1}{c}{$a_0$}
 & \multicolumn{1}{c}{$a_3$}
 & \multicolumn{1}{c}{$a_4$}\\\cline{3-7}
\multirow{2}*{}
& $\omega_1$
& $(1,4)$
& $(4,2)$
& $(0,\frac12)$
& $(3,-2)$
& $(0,-8)$\\\cline{3-7}
& $\omega_2$
& $(2,-8)$
& $(-1,-2)$
& $(0,\frac12)$
& $(3,2)$
& $(0,4)$\\\cline{3-7}
\end{tabular}
\end{table}

The receiver's best response is $a_1$ for
$p\in[0,\frac14]$, $a_2$ for
$p\in[\frac14,\frac38]$, $a_0$ for
$p\in[\frac38,\frac58]$, $a_3$ for
$p\in[\frac58,\frac34]$, and $a_4$ for
$p\in[\frac34,1]$.

In every odd period, the sender induces belief
$p^{\mathrm{o}}_{\ell}=0$ and $p^{\mathrm{o}}_{h}=1$.
In the following even period, she repeats the previous message,
inducing beliefs $p^{\mathrm{e}}_{\ell}=\frac13$ and $p^{\mathrm{e}}_{h}=\frac23$,
where $p^{\mathrm{e}}_m=p^{\mathrm{o}}_mQ$ for $m\in\{\ell,h\}$.

The receiver best responds to these beliefs. In odd periods, he takes
$a_1$ after $\ell$ and $a_4$ after $h$. In even periods, he takes
$a_2$ after $\ell$ and $a_3$ after $h$. He follows this strategy as
long as the message in every even period coincides with the message
sent in the preceding odd period, so that only the pairs $(\ell,\ell)$ and $(h,h)$
occur. If this condition is violated, he ignores all subsequent
messages and plays $a_0$ forever.

We first verify that the sender does not want to switch the odd-period
messages while preserving the continuation pattern. Under the prescribed
strategy, her expected payoff is $1/2$ in odd periods and $8/3$ in even
periods. Thus, her undiscounted long-run payoff equals $\frac12\left(\frac12+\frac83\right)=\frac{19}{12}$.

If she swaps the odd-period messages and then repeats the chosen message in
the following even period, her expected payoff is $1$ in odd periods. In
even periods, the receiver takes $a_3$ at belief $1/3$ and $a_2$ at belief
$2/3$, giving the sender expected payoff $\frac12\left(3+\frac23\right)=\frac{11}{6}$. 
Her long-run payoff from this deviation is therefore $\frac12\left(1+\frac{11}{6}\right)
=
\frac{17}{12}
<
\frac{19}{12}.$ The same comparison rules out all other deviations that preserve the odd--even pairing.\footnote{It is enough to consider pure deviations. Besides switching the two odd-period messages, the remaining pure deviations pool the two states: the sender sends the same message regardless of the odd-period state and repeats it in the following even period. This results in a long-run payoff of $3/2$.} Finally, any deviation that breaks the odd--even pairing triggers $a_0$ forever and is therefore unprofitable.  Since these inequalities are strict, they continue to hold for all sufficiently patient players.

On the other hand, the receiver has no profitable deviation because he always takes a myopic best response on the equilibrium path. The prescribed strategy profile is therefore a uniform equilibrium with the corresponding limit payoff 
\[
e=
\left(\frac{19}{12},\frac73\right).
\]

We now show that this payoff cannot be obtained as a convex
combination of static partial-commitment equilibrium payoffs. First, no static equilibrium can induce a posterior strictly below $1/4$. Suppose, to the contrary, that some message $m$ induces $p_m<1/4$. The receiver then uniquely chooses $a_1$ after $m$. Because the average posterior equals $1/2$, the other message $m'$ must induce a posterior $p_{m'}>1/2$. After $m'$, the receiver chooses among $a_0,a_3$, and $a_4$.

For sufficiently small $\varepsilon>0$, consider the alternative posterior profile $\widetilde p_m=p_m+\varepsilon$ and $\widetilde p_{m'}=p_{m'}-\frac{\lambda(m)}{\lambda(m')}\varepsilon$. This perturbation preserves Bayes plausibility and the message probabilities. Under $a_1$, the sender's payoff difference between states is
\[
u_S(\omega_2,a_1)-u_S(\omega_1,a_1)=1,
\]
whereas under each action $a\in\{a_0,a_3,a_4\}$,
\[
u_S(\omega_2,a)-u_S(\omega_1,a)=0.
\]
Hence, the perturbation strictly increases the sender's payoff, contradicting her equilibrium condition. Therefore, every posterior induced with positive probability in a static equilibrium satisfies $p_m\geq\frac14$.

For every $p\in[\frac14,1]$, the receiver's optimal payoff satisfies
\[
\max_{a\in A}\sum_{\omega\in\Omega}p(\omega)u_R(\omega,a)
=
\max\left\{
4-12p,\,
2-4p,\,
\frac12,\,
-2+4p,\,
-8+12p
\right\}
\leq 4p.
\]
Therefore, the receiver's payoff in any static partial-commitment equilibrium is at most
\[
e_R
\leq
4\sum_m\lambda(m)p_m
=
4\mu
=
2.
\]
The same bound holds for every convex combination of static
equilibrium payoffs. Since the dynamic equilibrium gives the receiver
payoff $\frac73>2,$ we conclude that
\[
e
\notin
\operatorname{Co} 
\mathcal E(\mu).
\]
\end{example}

This example shows that the set of undetectable deviations in the dynamic game may be strictly smaller than the set of deviations that preserve the marginal distribution. We therefore focus on a class of deviations that are always feasible and undetectable, yielding a necessary condition.

Copula deviations form such a class. They correspond to applying the equilibrium strategy to a fictitious state process whose joint distribution with the true state process is described by a copula with marginal \(\mu\). As shown by \citet{rsv}, in pseudo-renewal Markov chains such fictitious processes can be constructed so as to have the same law as the true state process, making the resulting deviations undetectable.

Let
\[
\mathcal C(\pi)
=
\left\{
c \in \Delta(\Omega \times \Omega):
\sum_{\omega' \in \Omega} c(\omega,\omega')=\pi(\omega)
\text{ and }
\sum_{\omega \in \Omega} c(\omega,\omega')=\pi(\omega')
\right\}
\]
denote the set of copulas with marginals \(\pi\) and let $c(\omega\mid\omega')=\frac{c(\omega,\omega')}{\pi(\omega')}$ whenever $\pi(\omega')>0$.

Given an outcome \(\nu\in \Delta(\Omega \times A)\) with prior \(\pi \in \Delta(\Omega)\), define the set of copula deviations 
\[
\mathcal M(\nu)
=
\left\{
\tilde\nu \in \Delta(\Omega \times A):
\tilde\nu(\omega,a)
=
\sum_{\omega' \in \Omega} c(\omega \mid \omega')\,\nu( \omega',a)
\text{ for some } c \in \mathcal{C}(\pi)
\right\}.
\]

\begin{definition}
An outcome $\nu \in \Delta(\Omega \times A)$   is \textbf{robust to copula deviations} if
\[
\EE_{\nu}[u_S]\geq \EE_{\tilde\nu}[u_S] \quad \text{for every } \tilde\nu \in \mathcal M(\nu).
\]
\end{definition}

Finally, we impose the receiver’s incentive constraint. An outcome $\nu \in \Delta(\Omega \times A)$ is \textbf{obedient} if, for every $a,a' \in A$,
\[
\sum_{\omega \in \Omega} \nu(\omega,a)\Bigl(u_R(\omega,a)-u_R(\omega,a')\Bigr)\ge 0.
\]
Equivalently, each recommended action is a myopic best response to the posterior it induces.

The next result shows that any uniform equilibrium payoff must be generated by an outcome that is both obedient and robust to copula deviations.

\begin{proposition}\label{prop:necessary}
For every uniform equilibrium payoff \(e\), there exists an outcome
\(\nu \in \Delta(\Omega \times A)\) with prior \(\mu\) such that
\(e = \mathbb{E}_\nu[(u_S,u_R)]\) and \(\nu\) is obedient and robust to copula deviations.
\end{proposition}

\Cref{prop:necessary} is not sufficient in general, since $\mathcal{M}(\nu)$ may be strictly smaller than the set of deviations in $\Sigma(\pi,\lambda)$.\footnote{To see this, suppose $\rho(m \mid \omega) > 0$ for all $\omega,m$, so that each posterior $p_m$ has full support. Copula deviations  preserve support and therefore cannot generate degenerate beliefs, whereas such posteriors are feasible in $\Sigma(\pi,\lambda)$.} One  implication is that any outcome that induces a uniform equilibrium payoff corresponds to a Bayes correlated equilibrium \citep{bergemann2016bayes} with prior $\mu$, since, in the limit, each period generates a Bayes plausible distribution of beliefs with prior $\mu$ that satisfies obedience.

In \cite{rsv}, where the sender reports states and observes the receiver's actions, copula deviations are both necessary and sufficient to characterize equilibrium outcomes for pseudo-renewal chains. In our setting, however, the sender sends messages rather than reporting states, so deviations that preserve the marginal distribution of messages may still alter the temporal structure. Thus, while copula deviations remain necessary, they are not sufficient to characterize equilibrium outcomes, as illustrated in Example \ref{ex:counter-example}.

Finally, we consider the special case where the states are drawn i.i.d. with distribution $\mu$. In this case, the intertemporal correlations discussed above disappear, and the static partial-commitment benchmark becomes exact.

\begin{proposition}\label{prop:iid}
Suppose that the state process is i.i.d. with distribution $\mu$. Then the set of uniform equilibrium payoffs coincides with  the convex hull of equilibrium payoffs of the persuasion model with partial commitment:
\[
\mathcal{U}=\mathrm{Co} \, \mathcal{E}(\mu).
\]
\end{proposition}

The inclusion follows from \Cref{th:main}. For the converse, in each period the sender's strategy induces a Bayes plausible distribution over posterior beliefs with weights \(\lambda\). Since the receiver's obedience condition holds in the limit, if the induced payoff does not belong to \(\mathcal E(\mu,\lambda)\), then the sender must have a profitable deviation preserving \(\lambda\). Since states are drawn i.i.d., such a deviation does not affect the distribution of message sequences and is therefore undetectable and can be implemented independently across periods, yielding a contradiction. Therefore, the induced payoff must belong to \(\mathcal E(\mu,\lambda)\), and the overall payoff is a convex combination of such payoffs across marginals \(\lambda\).

\section{Discussion}
\label{sec:discussion}

\subsection{Uniform equilibrium and Perfect Bayesian equilibrium}
\label{sec:uniform-pbe}

The strategy profile constructed in \Cref{th:main} does not constitute a Perfect Bayesian equilibrium (PBE).\footnote{Uniform equilibrium and PBE are not ordered by inclusion. PBE imposes exact sequential rationality after every history at a fixed discount factor, whereas uniform equilibrium requires approximate optimality uniformly for all sufficiently high discount factors.} The key issue is the receiver's sequential rationality. Since the receiver's current action does not affect future play, PBE requires him, after every history, to take a myopic best response to his belief.

However, if the quota for message \(m\) has already been reached within a block, the receiver uses the response rule corresponding to a different message \(\tilde m\neq m\). At such histories, the receiver need not take a myopic best response. Hence, the strategy profile is not sequentially rational for the receiver. This failure has a negligible impact on the discounted payoff: the fraction of periods in which quota replacement occurs, and hence in which a suboptimal action may be taken, vanishes as the block length grows. The receiver's gain from always taking myopic best responses can therefore be made arbitrarily small for sufficiently patient players. Moreover, as \Cref{prop:necessary} shows, every uniform equilibrium payoff is generated by an outcome that satisfies receiver obedience in the limit.

Uniform equilibrium captures a natural form of statistical discipline: the receiver holds the sender accountable for the empirical distribution of her messages. If the receiver always took a myopic best response, this disciplining device would disappear. These vanishing departures from sequential rationality are what allow the dynamic game to generate equilibrium payoffs that differ from those in \cite{KUVALEKARetal}. In particular, the sender can attain the Bayesian persuasion payoff when her payoff is state-independent.

\subsection{On the pseudo-renewal assumption}
\label{sec:pseudo-renewal}

The pseudo-renewal assumption ensures that the equilibrium posterior beliefs can be induced recursively, period after period. This assumption is sufficient, but not necessary, for the implementation. More generally, for any transition matrix $Q$, the same construction works whenever there exists a stochastic matrix $W$ such that
\[
p_mQ=\sum_{r\in M}W_{mr}p_r
\quad\text{for every }m\in M, \text{ and }
\lambda W=\lambda.
\]
The first condition ensures that the posterior collection is closed under the Markov update, while the second ensures that the induced transition preserves the message distribution. The latter condition is automatic when the  posteriors are affinely independent, since Bayes plausibility and the invariance of $\mu$ imply that the corresponding weights are unique.  The closure condition is   closely related to the absorbing condition of
\cite{Lehrer2026}.\footnote{A non-empty set $X\subseteq \Delta(\Omega)$ is
$Q$-absorbing if $pQ\in \mathrm{Co} \, X$ for every $p\in X$.}

The incentive argument is unchanged: the quota mechanism ensures that any sender
deviation in the dynamic game corresponds to a payoff from a deviation in the
static persuasion model with partial commitment. Some
deviations in the static model that preserve the marginal distribution over
messages may fail to be closed under the Markov update $Q$ and therefore need not
be feasible dynamically. This only reduces the sender's feasible deviations; hence,
if the sender has no profitable deviation in the static model, she has no
profitable deviation in the dynamic game.

The role of the pseudo-renewal assumption is to make this closure property uniform.  If every finite Bayes plausible collection of posteriors with prior $\mu$ is closed under the Markov update, then the chain must be pseudo-renewal; see \Cref{app:pseudo-renewal-necessity}.

\section{Conclusion }
\label{sec:conclusion}
We studied a dynamic sender-receiver game in which the sender receives no feedback about the receiver's action. The paper provides a micro-foundation for persuasion with partial commitment by showing how its key restriction can arise endogenously. In the static benchmark, the sender can deviate only through policies that preserve the marginal distribution over messages. In the dynamic game, this restriction is enforced by the receiver's monitoring of message frequencies.

This yields a partial characterization of uniform equilibrium payoffs, which becomes exact when states are i.i.d. In particular, when the sender's payoff is state-independent, dynamic interaction closes the commitment gap and allows the sender to attain the Bayesian persuasion payoff.

Several directions for future research remain. First, a full characterization of uniform
equilibrium payoffs remains open. Second, one could study settings in which the sender
receives imperfect feedback about the receiver's actions. Third, it would be interesting to
examine the role of richer statistical tests, such as tests based on higher-order frequencies
of messages (see \citealp{escobar2013efficiency,renou2015approximate,horner2017keeping}).
Finally, extending the analysis to more general dynamic environments is a promising
direction.

\begin{appendix} 
\section{Omitted Proofs}\label{proofs}

\subsection{Proof of \Cref{th:main}}
First, we prove that any equilibrium payoff of the  static persuasion model with partial commitment is a uniform equilibrium payoff. Then we show that any convex combination of these payoffs across message distributions is also a uniform equilibrium payoff.

Fix $e \in \mathcal{E}(\mu,\lambda)$. Let $(\sigma^*,\tau^*)$ denote the strategy profile constructed in the proof sketch. The proof proceeds in two steps. First, we show that the payoff induced by $(\sigma^*,\tau^*)$ can be made arbitrarily close to $e$ for sufficiently patient players. Second, we show that for any $\epsilon>0$, the strategy profile $(\sigma^*,\tau^*)$ is an $\epsilon$-equilibrium for all sufficiently patient players.  Recall that play is divided into blocks of length $N$.

\begin{lemma}\label{lemma:convergence}
For every $\epsilon>0$, there exists $N_0 \in \mathbb{N}$ and $\delta_0<1$ such that, for every $N\geq N_0$ and $\delta \geq \delta_0$, the strategy profile $(\sigma^*,\tau^*)$ induces a payoff within $\epsilon$ of $e$.
\end{lemma}

\begin{proof}

First, we show that at any period $n$, the expected joint distribution over state $\omega$ and message $m$ equals $\lambda(m)p_m(\omega)$. The claim follows by induction. It is immediate in the first period from the static signaling rule. If it holds at period $n-1$, then after message $\tilde m$ the next-period belief is $q_{\tilde m}=p_{\tilde m}Q$, so
\[
\PP_{\sigma^*,\tau^*}(m_{n-1}=\tilde m,\omega_n=\omega)
=
\lambda(\tilde m)q_{\tilde m}(\omega).
\]
Therefore,
\begin{align*} 
   \mathbb{E}_{\sigma^*}\!\left[\boldsymbol{1}_{\{\omega_n=\omega,m_n=m\}}\right]
   &=  
\sum_{\tilde m \in M}
\sigma^*_n(m_n=m \mid m_{n-1}=\tilde{m},\omega_n=\omega)
\PP_{\sigma^*, \tau^*}(m_{n-1}=\tilde{m}, \omega_n=\omega)\\
    &= \sum_{\tilde m \in M}
    \frac{p_m(\omega)}{q_{\tilde{m}}(\omega)}
    \big((1-\alpha)\lambda(m) + \alpha \boldsymbol{1}_{\{\tilde{m}=m \}}\big)
    q_{\tilde{m}}(\omega)\lambda(\tilde m)\\
    &= p_m(\omega)
    \sum_{\tilde m \in M}
    \big((1-\alpha)\lambda(m) + \alpha \boldsymbol{1}_{\{\tilde{m}=m \}}\big)
    \lambda(\tilde m)\\
    &= \lambda(m) p_m(\omega).
\end{align*}

If $\lambda(m)>0$, define $$ \overline{p}_{m}(\omega):=  \EE_{\sigma^*,\tau^*} \Big[ \frac{1}{N \lambda(m)} \sum_{t=1}^N \boldsymbol{1}_{\{\omega_t=\omega, r_t=m \}}  \Big].$$ 
So, \(\lambda(m)\overline p_m(\omega)\) is the expected fraction of periods in a block in which the state is \(\omega\) and the message used is \(m\).

For a fixed block length $N$, only the quota counters are reset at the beginning of each block. Since the initial distribution is invariant and the sender's strategy is recursive across periods, the expected play is the same from block to block. Hence, as $\delta\to 1$, the discounted payoff of player $i\in\{S,R\}$ converges to the expected average payoff over one block.
\begin{equation} \label{eq:limit}
 \lim_{\delta \to 1} \gamma^\delta_i(\sigma^*,\tau^*)=\sum_{m \in M} \lambda(m) \sum_{\omega \in \Omega} \overline{p}_m(\omega) \sum_{a \in A} \kappa(a \mid m) u_i(\omega,a).
\end{equation}

We will show that for sufficiently large block length, this discounted payoff vector can be made arbitrarily close to the target $e$.

First, by a law of large numbers for Markov chains, for sufficiently large \(N\), the quota of each message is exhausted only in a small fraction of periods with high probability. Hence, under \((\sigma^*,\tau^*)\), the expected fraction of periods in which the message sent differs from the message used can be made arbitrarily small. It follows that \(\overline p_m\) can be made arbitrarily close to \(p_m\) for every \(m \in M\).

 Using the fact that payoffs are bounded and the linearity of the expected payoff on the right-hand side of \eqref{eq:limit}, it follows that for every $\epsilon>0$, there exists $N_0 \in \mathbb{N}$ and $\delta_0 \in (0,1)$ such that for all $N \geq N_0$ and $\delta \geq \delta_0$, it holds that $$ \|\gamma^\delta(\sigma^*,\tau^*) - e\|_\infty \leq \epsilon.$$
 \end{proof}

We now show that the strategy profile $(\sigma^*,\tau^*)$ is an  $\epsilon$-equilibrium for sufficiently patient players. 

\begin{lemma}\label{lemma:equilibrium}
For every $\epsilon > 0$, there exists $N_0 \in \mathbb{N}$ and $\delta_0 \in (0,1)$ such that for all $N \geq N_0$ and $\delta \ge \delta_0$, the strategy profile $(\sigma^*,\tau^*)$ is an $\epsilon$-equilibrium.
\end{lemma}

\begin{proof}
Fix $\epsilon>0$. We show that there exist $N_0$ and $\delta_0$ such that for all $N\geq N_0$ and all $\delta\geq \delta_0$, no player can gain more than $\epsilon$ by deviating from $(\sigma^*,\tau^*)$.

\textbf{Sender's deviations.} Since quota counters are reset at the beginning of each block, any sender deviation induces in each block a collection of beliefs with weights $\lambda$. These collections are Bayes-plausible with prior $\mu$, so the sender's static incentive constraint applies block by block.

Fix any deviation $\tilde{\sigma}$. Since under $\tau^*$ the receiver plays according to message $m$ exactly $N\lambda(m)$ times in each block, define
\[
\tilde p_m(\omega)
:=
\frac{1}{N\lambda(m)}
\mathbb{E}_{\tilde{\sigma},\tau^*}
\left[
\sum_{n=1}^N \boldsymbol{1}_{\{\omega_n=\omega,r_n=m\}}
\right].
\]
Then $(\tilde p_m)_{m\in M}$ satisfies Bayes plausibility with weights $\lambda$, that is, $\mu=\sum_{m} \lambda(m)\tilde p_m.$ So, the collection $(\tilde p_m)_{m\in M} \in \Sigma(\mu,\lambda)$.

For fixed $N$, as $\delta\to 1$, the sender's discounted payoff under $(\tilde{\sigma},\tau^*)$ converges to the expected average payoff within a block. So, for $\delta$ sufficiently close to $1$,
\[
\gamma_S^\delta(\tilde{\sigma},\tau^*)
\le
\sum_{m\in M} \lambda(m)\sum_{\omega\in\Omega}\tilde p_m(\omega) \sum_{a \in A} \kappa( a \mid m) u_S(\omega,a)
+\frac{\epsilon}{2}.
\]
Since $(p_m)_{m\in M}$ has no profitable deviation within $\Sigma(\mu,\lambda)$, we have
\[
\sum_{m\in M} \lambda(m)\sum_{\omega\in\Omega}\tilde p_m(\omega)\sum_{a \in A} \kappa(a \mid m) u_S(\omega,a)
\le
e_S.
\]
Combining the two inequalities gives
\[
\gamma_S^\delta(\tilde{\sigma},\tau^*)
\le e_S+\frac{\epsilon}{2}.
\]
By \Cref{lemma:convergence}, for $N$ large enough and $\delta$ sufficiently close to $1$,
\[
\gamma_S^\delta(\sigma^*,\tau^*)\ge e_S-\frac{\epsilon}{2}.
\]
Therefore, for sufficiently large $N$ and $\delta$ close to $1$,
\[
\gamma_S^\delta(\tilde{\sigma},\tau^*)
-
\gamma_S^\delta(\sigma^*,\tau^*)
\le \epsilon.
\]

\textbf{Receiver's deviations.} Because the receiver's actions are unobservable, his choice in period $n$ does not affect the sender's future behavior. Hence, against $\sigma^*$, a myopic best response at every period is optimal. Let $\tau^{\mathrm{myp}}$ denote the myopic best-response strategy in which, after every history, the receiver best responds to his current belief $p_n$.

It therefore suffices to show that deviating to $\tau^{\mathrm{myp}}$ does not result in a gain greater than $\epsilon$. By construction, $\tau^*$ plays a myopic best response to the sender's message $m$ whenever the quota for $m$ is not exhausted.  By the law of large numbers for Markov chains, for $N$ large enough the empirical frequency of messages sent is close to $\lambda$ with high probability. Hence, $\tau^*$ coincides with $\tau^{\mathrm{myp}}$ in all but a small fraction of periods in each block. Since stage payoffs are bounded, the resulting difference in the  discounted payoffs can be made smaller than $\epsilon$ by choosing $N$ sufficiently large. Thus,  for sufficiently large $N$ and $\delta$ close to $1$,
\[
\gamma_R^\delta(\sigma^*,\tau^{\mathrm{myp}})
-
\gamma_R^\delta(\sigma^*,\tau^*)
\le \epsilon.
\]
Therefore, the receiver cannot gain more than $\epsilon$ by deviating.
\end{proof}

\Cref{lemma:convergence} and \Cref{lemma:equilibrium} show that any $e \in \mathcal{E}(\mu,\lambda)$ is a uniform equilibrium payoff. We now show that any convex combination of such payoffs is also attainable.

Pick any $e \in \mathrm{Co}\,\mathcal E(\mu)$. Since payoffs lie in a finite-dimensional space, $e$ can be written as a finite convex combination. Thus, there exists an integer $n$, weights $(\eta_i)_{i=1}^n$ with $\eta_i\ge 0$ and $\sum_{i=1}^n \eta_i=1$, and marginals $\lambda_i\in\Delta(M)$ such that
\[
e=\sum_{i=1}^n \eta_i e_i,
\qquad\text{with } e_i\in \mathcal E(\mu,\lambda_i)\ \text{for each } i=1,\dots,n.
\]

The players use a block strategy. Each block consists of $n$ active sub-blocks, separated by transition phases. In active sub-block $i$, the players play according to the strategy profile $(\sigma_i^*,\tau_i^*)$ associated with $e_i$ and $\lambda_i$.

Before each active sub-block, we insert a transition phase of length $\tilde N$. During this phase, the sender sends an uninformative fixed message and the receiver ignores messages, playing a myopic best reply to his current belief. Since the Markov chain is irreducible and aperiodic, $\tilde N$ can be chosen large enough so that, at the beginning of the next active sub-block, the receiver's belief is arbitrarily close to $\mu$. Hence, the desired posterior beliefs can again be induced recursively, up to an arbitrarily small error. Taking the active sub-blocks sufficiently long relative to $\tilde N$, the discounted weight of transition phases is negligible.

Choose the lengths of the active sub-blocks so that active sub-block $i$ occupies approximately a fraction $\eta_i$ of the active part of the block. Using the same argument as in \Cref{lemma:convergence}, for each active sub-block $i$, taking its length sufficiently large ensures that the payoff induced by $(\sigma_i^*,\tau_i^*)$ is arbitrarily close to $e_i$ for all $\delta$ sufficiently close to $1$. Since the transition phases have negligible discounted weight, it follows that the discounted payoff induced by the overall block strategy can be made arbitrarily close to
\[
\sum_{i=1}^n \eta_i e_i.
\]

It remains to check incentives. Since the transition phases have negligible discounted weight, it is enough to consider deviations during the active sub-blocks. Because the quota counters are reset at the beginning of each active sub-block, deviations can be evaluated separately within active sub-blocks. A deviation may affect the belief at the start of the following transition phase, but after the transition phase the desired posteriors can again be induced, up to an arbitrarily small error.

By \Cref{lemma:equilibrium}, for every $i=1,\dots,n$, there exists $\delta_0^i$ such that for every $\delta \ge \delta_0^i$, no player can gain more than an arbitrarily small amount by deviating in active sub-block $i$. Let $\delta_0=\max_{i=1,\dots,n}\delta_0^i$. Taking the active sub-blocks sufficiently long relative to $\tilde N$, and then taking $\delta$ sufficiently close to $1$, the total discounted gain from any deviation is at most $\epsilon$.

Using the convergence result, this implies that the difference in payoffs between the equilibrium strategy and any deviation can be bounded above by $\epsilon$. Therefore,
\[
\sum_{i=1}^n \eta_i e_i\in \mathcal{U}.
\]

\subsection{Proof of \Cref{prop:necessary}}

Consider the deviation \(\tau^{\mathrm{myp}}\) where the receiver plays a myopic best response to belief \(p_n\) at every history. Since continuation payoffs do not depend on the receiver's action, we have
\begin{equation} \label{eq:receiver-myopic-deviation}
\gamma^\delta_R(\sigma^*,\tau^{\mathrm{myp}}) - \gamma^\delta_R(\sigma^*,\tau^*)
=
\mathbb{E}_{\sigma^*,\tau^*}\Bigg[
\sum_{n=1}^\infty (1-\delta)\delta^{n-1}
\Big(
\mathbb{E}_{p_n}[u_R(\cdot,a^*(p_n))]
-
\mathbb{E}_{p_n}[u_R(\cdot,a_n)]
\Big)
\Bigg].
\end{equation}
Since \((\sigma^*,\tau^*)\) is an \(\epsilon\)-equilibrium, the left-hand side is at most \(\epsilon\).

Fix any \(\eta>0\), and let
\[
R_\eta^n
=
\left\{
\mathbb{E}_{p_n}[u_R(\cdot,a^*(p_n))]
-
\mathbb{E}_{p_n}[u_R(\cdot,a_n)]
\ge \eta
\right\}
\]
denote the event when the receiver's gain in period $n$ is at least $\eta$.

Then \eqref{eq:receiver-myopic-deviation} implies
\begin{align*}
\epsilon
&\ge
\mathbb{E}_{\sigma^*,\tau^*}\Bigg[
\sum_{n=1}^\infty (1-\delta)\delta^{n-1}
\Big(
\mathbb{E}_{p_n}[u_R(\cdot,a^*(p_n))]
-
\mathbb{E}_{p_n}[u_R(\cdot,a_n)]
\Big)
\Bigg] \\
&\ge
\eta\,
\mathbb{E}_{\sigma^*,\tau^*}\Bigg[
\sum_{n=1}^\infty (1-\delta)\delta^{n-1}\mathbf{1}_{R_\eta^n}
\Bigg].
\end{align*}
Hence,
\[
\mathbb{E}_{\sigma^*,\tau^*}\Bigg[
\sum_{n=1}^\infty (1-\delta)\delta^{n-1}\mathbf{1}_{R_\eta^n}
\Bigg]
\le \frac{\epsilon}{\eta}.
\]
Therefore, for every \(\eta>0\), the discounted weight of periods in which the receiver's action is worse than a myopic best response by at least \(\eta\) vanishes as \(\epsilon\to 0\). Equivalently, the induced discounted outcome satisfies the obedience inequalities up to a vanishing error.

Fix $\epsilon>0$, and let $(\sigma^*,\tau^*)$ be an $\epsilon$-equilibrium of the $\delta$-discounted game for all $\delta\ge \delta_0$. Let
\begin{equation*}
\nu(\omega, a)=  (1-\delta)\sum_{n=1}^\infty \delta^{n-1}  \mathbb{E}_{ \sigma^*, \tau^* }[\boldsymbol{1}_{\{\omega_n=\omega, a_n=a\}}]
\end{equation*}
denote the discounted expected outcome induced by the  strategy profile. Since the initial state is drawn from the invariant distribution, it follows that \(\sum_a \nu(\omega,a)=\mu(\omega)\).

For the sender, we construct an undetectable deviation $\sigma'$. Since $(\sigma^*,\tau^*)$ is an $\epsilon$-equilibrium, this deviation can improve the sender's payoff by at most $\epsilon$. Using Lemma 4 in \cite{rsv}, in pseudo-renewal Markov chains, one can construct a sequence of fictitious states $(\xi_n)$ that is statistically indistinguishable from the sequence $(\omega_n)$. In particular, (i) the sequence $\{\xi_n\}_n$ has the same law  as $\{\omega_n\}_n$, (ii) in each period $n$,  the law  of the pair $(\omega_n,\xi_n)$ is given by a distribution $c \in \mathcal{C}(\mu)$ and (iii) the conditional law of $\omega_n$ given $\xi_1,\ldots,\xi_n$ is given by $c(\cdot \mid \xi_n)$. Moreover, the sequence $\xi_n$ can be constructed using information available only in period $n$.

In any period $n$, the sender sends a message according to $\sigma^*$ evaluated at the fictitious history $(\xi_1,m_1,\ldots,\xi_n)$. Formally, given a history $(\omega_1,\xi_1,m_1,\ldots,\omega_n,\xi_n)$ consisting of realized and fictitious states and messages, the strategy $\sigma'$ selects the current message according to the distribution that $\sigma^*$ would assign after the fictitious history $(\xi_1,m_1,\ldots,\xi_n)$.

We now show that the sender's expected payoff under the strategy profile
\((\sigma',\tau^*)\) is equal to
\(\sum_{\omega,\xi,a} c(\omega\mid\xi)\nu(\xi,a)u_S(\omega,a)\). For each period \(n\),
\begin{align*}
    \mathbb{E}_{\sigma',\tau^*}[u_S(\omega_n,a_n)]
    &=
    \sum_{\omega_n}\sum_{a_n}
    \mathbb{P}_{\sigma',\tau^*}(\omega_n,a_n)u_S(\omega_n,a_n)\\
    &=
    \sum_{\omega_n}\sum_{\xi_1,\ldots,\xi_n}\sum_{m_1,\ldots,m_n}\sum_{a_n}
    \mathbb{P}_{\sigma',\tau^*}(\xi_1,m_1,\ldots,\omega_n,\xi_n,m_n,a_n)
    u_S(\omega_n,a_n)\\
&=
\sum_{\omega_n}\sum_{\xi_1,\ldots,\xi_n}\sum_{m_1,\ldots,m_n}\sum_{a_n}
\mathbb{P}_{\sigma',\tau^*}(\omega_n \mid \xi_1,\ldots,\xi_n,m_1,\ldots,m_n) \\
&\qquad\qquad\qquad \times
\mathbb{P}_{\sigma',\tau^*}(\xi_1,m_1,\ldots,\xi_n,m_n,a_n)
u_S(\omega_n,a_n)\\
    &=
    \sum_{\omega_n}\sum_{\xi_n}\sum_{a_n}
    c(\omega_n\mid\xi_n)\,
    \mathbb{P}_{\sigma',\tau^*}(\xi_n,a_n)\,
    u_S(\omega_n,a_n),
\end{align*}
where the last equality uses property (iii). Since the message history is generated from the fictitious history under \(\sigma'\), it contains no additional information about \(\omega_n\) conditional on \((\xi_1,\ldots,\xi_n)\). Hence, $\mathbb{P}_{\sigma',\tau^*}(\omega_n \mid \xi_1,\ldots,\xi_n, m_1,\ldots,m_n)
=
c(\omega_n \mid \xi_n)$.

Since the process \(\{\xi_n\}_{n \in \mathbb{N}}\) has the same law as \(\{\omega_n\}_{n \in \mathbb{N}}\), and since \(\sigma'\) applies \(\sigma^*\) to the fictitious history, the discounted joint distribution over \((\xi_n,a_n)\) under \((\sigma',\tau^*)\) coincides with the discounted joint distribution of \((\omega_n,a_n)\) under \((\sigma^*,\tau^*)\). Hence,
\[
(1-\delta)\sum_{n=1}^\infty \delta^{n-1}
\mathbb P_{\sigma',\tau^*}(\xi_n=\xi,a_n=a)
=
\nu(\xi,a).
\]
Therefore,
\[
\gamma^\delta_S(\sigma',\tau^*)
=
\sum_{\omega \in \Omega} \sum_{\xi \in \Omega}
c(\omega\mid\xi) \sum_{a \in A}\nu(\xi,a)u_S(\omega,a).
\]

Since $(\sigma^*,\tau^*)$ is an $\epsilon$-equilibrium, this deviation can
increase the sender's payoff by at most $\epsilon$. Hence, for every
$c\in\mathcal C(\mu)$,
\[
\sum_{\omega,a}\nu(\omega,a)u_S(\omega,a)
\ge
\sum_{\omega,\xi,a}c(\omega\mid \xi)\nu(\xi,a)u_S(\omega,a)-\epsilon.
\]
Thus, $\nu$ is robust to copula deviations up to an error $\epsilon$.

Finally, consider a sequence $\epsilon_k\to0$. For each $k$, take an
$\epsilon_k$-equilibrium with discount factor $\delta_k$, where
$\delta_k\to1$, and let $\nu^k$ be the induced discounted outcome. By compactness of
$\Delta(\Omega\times A)$, a subsequence converges to some $\nu$. Since
obedience and robustness hold up to a vanishing error along the sequence,
they pass to the limit. Hence, $\nu$ is obedient and robust to copula deviations. Moreover, its state marginal is $\mu$. Finally, since $\nu^k\to\nu$ and the payoffs induced by
$\nu^k$ converge to the uniform equilibrium payoff $e$, we have $$e_i=\EE_\nu \Big[u_i(\omega,a)\Big],
\ \text{ for } i\in\{S,R\}.$$

\subsection{Proof of \Cref{prop:iid}}

The inclusion follows from \Cref{th:main}. We prove the converse.

Fix $\epsilon>0$, and let $(\sigma^*,\tau^*)$ be an $\epsilon$-equilibrium of the $\delta$-discounted game for all $\delta\ge \delta_0$. At any history $h_n=(m_1,\ldots,m_{n-1})$, the sender's strategy induces a marginal distribution $\lambda_{h_n}\in\Delta(M)$ over messages. Since the state process is i.i.d., the receiver's belief before observing the message is always equal to the prior $\mu$. Hence the induced posterior beliefs $(p^m_{h_n})_{m\in M}$ satisfy $\sum_{m\in M}\lambda_{h_n}(m)p^m_{h_n}=\mu$.

The receiver’s equilibrium condition follows as in \Cref{prop:necessary}. Since continuation payoffs do not depend on the receiver's current action, a myopic best response is optimal. Hence, for every $\eta>0$, the discounted weight of periods in which the receiver's payoff is more than $\eta$ below the payoff from a myopic best response is bounded by $\epsilon/\eta$, and thus vanishes as $\epsilon\to0$.

We now consider the sender's incentive constraint. Suppose that at history $h_n$ there exists a feasible one-shot deviation in $\Sigma(\mu,\lambda_{h_n})$ that improves the sender's current payoff by at least $\eta$. Such a deviation preserves the conditional distribution of the current message given $h_n$. Hence, applied history by history, it leaves the law of messages unchanged and is therefore undetectable. Moreover, by the i.i.d.\ assumption, the deviation does not affect continuation payoffs, since future states are independent of the current state.

Let $S_\eta^n$ be the event that in period $n$ such a feasible one-shot deviation exists. If the sender applies these deviations whenever they are available, then
\[
\gamma^\delta_S(\sigma',\tau^*)-\gamma^\delta_S(\sigma^*,\tau^*)
\ge
\eta \,
\mathbb E_{\sigma^*,\tau^*}\left[
\sum_{n=1}^\infty (1-\delta)\delta^{n-1}\mathbf 1_{S_\eta^n}
\right].
\]
Since $(\sigma^*,\tau^*)$ is an $\epsilon$-equilibrium, it follows that
\[
\mathbb E_{\sigma^*,\tau^*}\left[
\sum_{n=1}^\infty (1-\delta)\delta^{n-1}\mathbf 1_{S_\eta^n}
\right]
\le \frac{\epsilon}{\eta}.
\]
Thus, as $\epsilon\to0$, the discounted weight of periods in which the sender's static incentive constraint is violated by at least $\eta$ vanishes.

It follows that, up to a discounted error that vanishes with $\epsilon$, play in each period induces an outcome satisfying the receiver's best-response condition and the sender's static equilibrium condition for some marginal distribution $\lambda_{h_n}$. By compactness of the finite static model, these approximate outcomes are arbitrarily close to payoffs in $\mathcal E(\mu)$. Therefore, the discounted payoff is arbitrarily close to a convex combination of payoffs in $\mathcal E(\mu)$. Hence, every uniform equilibrium payoff belongs to $\mathrm{Co} \,\mathcal E(\mu)$.

\subsection{Uniform closure under the pseudo-renewal assumption}
\label{app:pseudo-renewal-necessity}

The following proposition shows that the pseudo-renewal assumption is equivalent to the following uniform closure property: every finite Bayes plausible collection of posterior beliefs with prior \(\mu\) is closed under the Markov update.

\begin{proposition}
Let \(Q\) be a  transition matrix with full-support invariant distribution \(\mu \in \mathrm{int} \, \Delta (\Omega)\). The following are equivalent.

\begin{enumerate}
\item For every finite collection of posterior beliefs \(\{p_m\}_{m\in M}\subseteq\Delta(\Omega)\) such that
\[
\mu\in \mathrm{Co} \, \{p_m:m\in M\},
\]
the collection is closed under the Markov update:
\[
p_mQ\in \mathrm{Co}\, \{p_r:r\in M\}
\qquad\text{for every }m\in M.
\]

\item The Markov chain is pseudo-renewal.
\end{enumerate}
\end{proposition}

\begin{proof}
Suppose first that condition 1 holds. Fix any belief \(p\in\Delta(\Omega)\). For \(0<\lambda<1\), define
\[
q_\lambda=\frac{\mu-\lambda p}{1-\lambda}.
\]
Since \(\mu\) has full support, \(q_\lambda\in\Delta(\Omega)\) for all sufficiently small \(\lambda>0\). Moreover, \(\mu=\lambda p+(1-\lambda)q_\lambda\), so the  collection \(\{p,q_\lambda\}\) is Bayes plausible with prior \(\mu\). By condition 1, this collection is closed under the Markov update, and hence
\[
pQ\in\mathrm{Co} \, \{p,q_\lambda\}.
\]
Letting \(\lambda\to0\), we obtain \(pQ\in\mathrm{Co} \, \{p,\mu\}\).

Hence, for every \(p\), there exists \(\alpha(p)\in[0,1]\) such that
\begin{equation*} 
pQ=\alpha(p)p+(1-\alpha(p))\mu.
\end{equation*}

We now show that \(\alpha(p)\) is constant. Since \(\mu Q=\mu\), writing \(p=\mu+v\), with \(\sum_{\omega}v(\omega)=0\), gives
\[
(\mu+v)Q=\alpha(p)(\mu+v)+(1-\alpha(p))\mu=\mu+\alpha(p)v,
\]
and hence \(vQ=\alpha(p)v\). Thus \(Q\) scales every zero-sum direction. By linearity, the scaling factor must be the same in all directions: otherwise, applying \(Q\) to the sum of two independent zero-sum directions would give a contradiction. Hence, there exists \(\alpha\in[0,1]\) such that
\[
pQ=\alpha p+(1-\alpha)\mu
\qquad\text{for every }p\in\Delta(\Omega).
\]

Taking \(p=\delta_{\omega}\), we get
\[
Q(\cdot\mid\omega)
=
\alpha\delta_{\omega}+(1-\alpha)\mu.
\]
Thus, for every \(\omega\neq\tilde\omega\), \(Q( \omega\mid\tilde\omega)=(1-\alpha)\mu(\omega)\), which depends only on the destination state \(\omega\). Hence \(Q\) is pseudo-renewal.

Conversely, if \(Q\) is pseudo-renewal, then \(pQ=\alpha p+(1-\alpha)\mu\) for some \(\alpha\in[0,1]\) and every \(p\). Let \(\{p_m\}_{m\in M}\) be any finite collection, and suppose that \(\mu=\sum_{r\in M}\lambda_r p_r\), where \(\lambda_r\geq 0\) and \(\sum_{r\in M}\lambda_r=1\). Then, for every \(m\),
\[
p_mQ
=
\alpha p_m+(1-\alpha)\sum_{r\in M}\lambda_r p_r
\in \mathrm{Co}\, \{p_r:r\in M\}.
\]
Thus, the collection is closed under the Markov update.
\end{proof}

\end{appendix}

\bibliographystyle{jpe}
\bibliography{JET}

\end{document}